\documentclass[11pt]{article}

\usepackage[margin=1.5in]{geometry}
\usepackage{amsmath,amssymb,amsthm}
\usepackage{hyperref}
\usepackage{enumitem}
\usepackage{mathtools}

\title{\textbf{Remarks on Legal Entanglement: No-Signaling, Local Operations, and Legal Updates}}
\author{%
  Mikołaj Sienicki\thanks{Polish--Japanese Academy of Information Technology, ul.~Koszykowa~86, 02--008 Warsaw, Poland, European Union}%
  \quad and \quad
  Krzysztof Sienicki\thanks{Chair of Theoretical Physics of Naturally Intelligent Systems (\(\mathbb{NIS}\)\textsuperscript{\textcopyright}), Lipowa~2/Topolowa~19, 05--807 Podkowa Leśna, Poland, European Union; E-mail: \texttt{niskrissienicki@gmail.com}}%
}
\date{\today}

\newtheorem{lemma}{Lemma}

\begin{document}
\maketitle

\begin{abstract}
Godfrey \& Sichelman propose a quantum-inspired formalism for modeling coupled legal relations and interpretations under the label \emph{legal entanglement} \cite{GodfreySichelman2025LegalEntanglement}. The idea is attractive: it offers a compact way to talk about how legal variables can become tightly coupled, and it suggests quantitative proxies for legal modularity and ``information cost.''  

The main technical problem is narrower. In the paper’s discussion of \emph{formulative entanglement}, legislation is treated as a local operation on subsystem $A$ that can change the reduced state of a distant entangled subsystem $B$ (so $\rho_B' \neq \rho_B$) even before any ``measurement'' at $B$, and this is framed as a departure from the quantum no-signaling constraint. In standard quantum mechanics, however, no-signaling does not rely on unitarity alone: it holds for \emph{all} local, trace-preserving operations (CPTP maps) \cite{NielsenChuang2010,Kraus1983,Wilde2017}.  

This note states the correct no-signaling result, pinpoints exactly where the mapping to quantum mechanics becomes inconsistent, and suggests a clean repair that preserves the legal intuition: treat legislation either as (i) a \emph{global} update of the rule/constraint structure (changing the admissible state space or observables), and/or (ii) an LOCC-style process (local operations plus public dissemination of authoritative classical information). Finally, we point out that an ``updating-first'' perspective developed in \emph{Algorithmic Idealism} \cite{Sienicki2024AI1,Sienicki2024AI2,Sienicki2025AI3} and its legal extension \cite{SienickiSienicki2025AILawI,SienickiSienicki2025AILawII} offers a natural vocabulary for keeping two layers distinct: (a) physical locality/no-signaling constraints, and (b) institutional or semantic constraint propagation in legal systems.
\end{abstract}

\section{Target and scope}
These remarks address a specific technical point in Godfrey \& Sichelman's \emph{Legal Entanglement} \cite{GodfreySichelman2025LegalEntanglement}: the treatment of \emph{formulative entanglement} and its relationship to the quantum no-signaling constraint.

The paper introduces three contexts for legal entanglement (interpretive, formulative, adjudicative) and uses standard quantum-information tools (tensor products, partial trace, entanglement measures) to represent dependencies among legal variables. As a \emph{quantum-inspired} modeling strategy, much of the formal apparatus can be read as defining and exploring a family of probabilistic/correlational models on legal variables.

The issue arises when the paper draws a \emph{strong} contrast with quantum mechanics: it suggests that legislation, modeled as an operator on one subsystem, can alter the reduced state of a remote entangled subsystem even absent measurement, and that this is a principled divergence from quantum no-signaling. The remainder of this note argues that the divergence is not required, and that (with a small refactoring) the legal intuitions can be kept while the quantum claims become fully standard.

\section{No-signaling for local operations: the QM baseline}
The no-signaling constraint in quantum theory is operational: choices made on $A$ cannot change the \emph{unconditional} outcome statistics available to an operationally separated party $B$. This statement does not require relativity; in relativistic settings the same operational separation is naturally realized by spacelike separation \cite{NielsenChuang2010,Wilde2017}.

No-signaling is often motivated using local \emph{unitary} evolution, but the correct statement is broader. In quantum information theory, the most general physically allowed local operation is a completely positive trace-preserving (CPTP) map (which includes open-system dynamics and \emph{unselective} measurements) \cite{Kraus1983,NielsenChuang2010,Wilde2017}.

\begin{lemma}[No-signaling under local CPTP maps]
Let $\rho_{AB}$ be any bipartite density operator on $\mathcal{H}_A \otimes \mathcal{H}_B$.
Let $\mathcal{E}_A$ be any CPTP map acting on subsystem $A$.
Define
\begin{align}
\rho_{AB}' &:= (\mathcal{E}_A \otimes \mathrm{Id}_B)(\rho_{AB}), \label{eq:rhoABprime_def}\\
\rho_B &:= \mathrm{Tr}_A(\rho_{AB}), \qquad
\rho_B' := \mathrm{Tr}_A(\rho_{AB}'). \label{eq:rhoB_def}
\end{align}
Then $\rho_B' = \rho_B$.
\end{lemma}

\begin{proof}
Write $\mathcal{E}_A$ in Kraus form \cite{Kraus1983,NielsenChuang2010}:
\begin{equation}\label{eq:kraus_form}
\mathcal{E}_A(X)=\sum_k A_k X A_k^\dagger, \qquad \sum_k A_k^\dagger A_k = I_A.
\end{equation}
Then
\begin{equation}\label{eq:rhoABprime_kraus}
\rho_{AB}' = \sum_k (A_k \otimes I_B)\,\rho_{AB}\,(A_k^\dagger \otimes I_B).
\end{equation}
Taking the partial trace over $A$ and using the basic identity
\begin{equation}\label{eq:partialtrace_cyclicity}
\mathrm{Tr}_A\!\left[(X_A\otimes I_B)Y_{AB}\right]
=
\mathrm{Tr}_A\!\left[Y_{AB}(X_A\otimes I_B)\right],
\end{equation}
we obtain
\begin{align}
\rho_B'
&= \mathrm{Tr}_A(\rho_{AB}') \nonumber\\
&= \sum_k \mathrm{Tr}_A\!\Big((A_k \otimes I_B)\,\rho_{AB}\,(A_k^\dagger \otimes I_B)\Big) \label{eq:trace_step1}\\
&= \sum_k \mathrm{Tr}_A\!\Big(\rho_{AB}\,(A_k^\dagger A_k \otimes I_B)\Big) \label{eq:trace_step2}\\
&= \mathrm{Tr}_A\!\Big(\rho_{AB}\,(\sum_k A_k^\dagger A_k \otimes I_B)\Big) \label{eq:trace_step3}\\
&= \mathrm{Tr}_A(\rho_{AB}\,(I_A \otimes I_B)) \label{eq:trace_step4}\\
&= \mathrm{Tr}_A(\rho_{AB})
= \rho_B. \label{eq:nosignaling_result}
\end{align}
\end{proof}

\subsection{Selective vs.\ unselective measurements (where ``remote updating'' belongs)}
A useful place where “remote updates’’ \emph{do} appear is when we condition on a measurement outcome. The key point is that there are two different operations \cite{NielsenChuang2010,Wilde2017,PreskillLectureNotes}:

\begin{itemize}[leftmargin=2em]
\item Unselective measurement: we measure but ignore the outcome. This is still a CPTP map, so it cannot change $\rho_B$ (Lemma~1 / \eqref{eq:nosignaling_result}).
\item Selective measurement: we \emph{condition} on a particular outcome $k$. This produces a conditional state $\rho_{B|k}$ that depends on $k$.
\end{itemize}

Concretely, let $\{M_k\}_k$ be Kraus operators on $A$ satisfying the completeness condition
\begin{equation}\label{eq:mk_completeness}
\sum_k M_k^\dagger M_k = I_A,
\end{equation}
so that the corresponding unselective measurement is trace-preserving on $A$. Define the outcome probability (full trace over $AB$)
\begin{equation}\label{eq:pk_def}
p_k=\mathrm{Tr}_{AB}\!\left[(M_k\otimes I_B)\,\rho_{AB}\,(M_k^\dagger\otimes I_B)\right],
\end{equation}
and the normalized conditional state
\begin{equation}\label{eq:conditional_state}
\rho_{B|k}=\frac{1}{p_k}\,\mathrm{Tr}_A\!\Big[(M_k\otimes I_B)\,\rho_{AB}\,(M_k^\dagger\otimes I_B)\Big].
\end{equation}
The crucial operational point is simple: $B$ can only \emph{use} this dependence on $k$ if the outcome label is communicated by classical means.
Otherwise, $B$ still sees the same unconditional marginal.
 
\section{Where the target paper's mapping breaks}
Godfrey \& Sichelman propose that, unlike standard quantum systems, a legislative enactment modeled as an operator on subsystem $A$ may alter the reduced state of an entangled subsystem $B$ prior to any legal ``measurement'' at $B$, and they describe this as a divergence from no-signaling \cite{GodfreySichelman2025LegalEntanglement}.

From the perspective of Lemma~1 (equations \eqref{eq:rhoABprime_def}--\eqref{eq:nosignaling_result}), there are only a few coherent ways to interpret such a claim:
\begin{enumerate}[leftmargin=2em]
\item $P_A$ is not trace-preserving (i.e.\ it is post-selection / conditioning).
Then $B$ can change conditionally as in \eqref{eq:conditional_state}, but without communication, $B$ still sees the same unconditional marginal \eqref{eq:nosignaling_result}.
\item $P_A$ is not a local operation on a fixed bipartite state.
For example, it may represent a \emph{global} update to the joint state assignment or to the permissible observables/constraints.
\item The model includes an explicit external communication channel.
If one posits a channel that propagates a signal ``outside physical space,'' then one has in effect posited communication (albeit in a different ontology).
\end{enumerate}

The simplest repair is to stop treating legislation as a local quantum operation on subsystem $A$ and instead treat it as a \emph{global rule update} or \emph{conditioning + dissemination} process.

\section{A QM-consistent reframing of ``formulative entanglement''}

\subsection{Option A: legislation as a global rule/constraint update}
In law, legislation changes the \emph{normative constraints} governing which states are admissible and which questions are legally meaningful. In a quantum-inspired model, this is most naturally represented not as a local map
\begin{equation}\label{eq:local_map_form}
\rho_{AB} \mapsto (\mathcal{E}_A \otimes \mathrm{Id}_B)(\rho_{AB}),
\end{equation}
but as a transformation of the \emph{model itself}, e.g.
\begin{equation}\label{eq:global_model_update}
(\rho_{AB}, \mathcal{O}, C) \mapsto (\rho_{AB}', \mathcal{O}', C'),
\end{equation}
where $\mathcal{O}$ encodes the set of legal ``observables'' (the partitions of outcomes a decision-maker can realize) and $C$ encodes intertextual and doctrinal constraints.

\subsection{Option B: legislation as LOCC-style conditioning + public dissemination}
A complementary reframing treats legal updating as an explicitly informational process: a legal act produces a public record, and legal actors update their state assignments upon learning the record.

In quantum information terms, this is structurally analogous to LOCC, i.e.\ \emph{Local Operations and Classical Communication}. LOCC is best thought of as a \emph{class of protocols}: possibly multi-round and adaptive sequences of local CPTP maps whose choice can depend on previously communicated classical outcomes, together with the dissemination of those outcome records/labels via classical channels \cite{NielsenChuang2010,Wilde2017,PreskillLectureNotes}. This is the standard quantum-information template for ``update at a distance without signaling.''

\section{Comment and related work: Algorithmic Idealism and Algorithmic Idealism of Law}
Two related ``updating-first'' strands help keep the layers separated.

First, the \emph{Algorithmic Idealism} sequence \cite{Sienicki2024AI1,Sienicki2024AI2,Sienicki2025AI3} treats ``collapse'' as informational updating of an agent’s state assignment and treats entanglement as a structured coupling constraint on joint predictions. On that stance, instantaneous change is naturally interpreted as a change in conditional distributions (given new information), not as an operational signaling channel.

Second, the \emph{Algorithmic Idealism of Law} short versions \cite{SienickiSienicki2025AILawI,SienickiSienicki2025AILawII} provide a directly legal instantiation of the same separation:
\begin{itemize}[leftmargin=2em]
\item In the single-system U.S.\ framing, authoritative acts (judgments, enactments, settlements) are modeled as ``legal measurements'' that trigger Bayesian updating over a predictive kernel on \emph{legal self-states} \cite{SienickiSienicki2025AILawI}. This matches precisely the conditional/epistemic reading of remote ``updates'' suggested by \eqref{eq:conditional_state}.
\item In the European multi-level framing, legally authoritative updates occur at multiple levels (EU/CJEU, ECHR/ECtHR\footnote{EU = European Union; CJEU = Court of Justice of the European Union; ECHR = European Convention on Human Rights; ECtHR = European Court of Human Rights.}, national orders), and coupling is represented via interface constraints between kernels \cite{SienickiSienicki2025AILawII}. This makes it especially clear why ``instantaneous'' legal effects should be modeled as constraint propagation in a layered institutional architecture rather than as a local physical operation that changes a distant reduced state.
\end{itemize}
Incorporating these references clarifies that the repaired account of ``formulative entanglement'' is not an ad hoc patch: it is the expected form of update in an informational jurisprudence.

\section{Conclusion}
The target paper's conceptual contribution---treating certain legal dependencies as entanglement-like and using quantum-information measures as quantitative diagnostics---does not require any departure from standard quantum mechanics \cite{GodfreySichelman2025LegalEntanglement}. The appearance of a no-signaling violation stems from mapping legislative change to a local subsystem operation on a fixed bipartite state. Recasting legislation as a global rule update \eqref{eq:global_model_update} and/or as LOCC-like conditioning plus dissemination \eqref{eq:conditional_state} yields a quantum-mechanically faithful account while preserving the legal phenomena the paper seeks to capture, and it aligns naturally with the algorithmic-state jurisprudence program \cite{SienickiSienicki2025AILawI,SienickiSienicki2025AILawII}.

\bibliographystyle{plain}
\bibliography{refs}

\end{document}